\documentclass{article}

\usepackage{hyperref}

\usepackage{graphicx, subfig}
\usepackage{amssymb}
\usepackage{amsmath,amsthm}
\usepackage{amsfonts,color}
\usepackage{natbib}

\newtheorem{assum}{\textit{Assumption}}
\newtheorem{prop}{\textit{Proposition}}
\newtheorem{corollaire}{\textit{Corollary}}

\newcommand{\ac}{\color{red}}   
\newcommand{\ca}{\rm \color{black}} 
\newcommand{\cutac}[1]{\ac [cut] \ca}            

\date{10 November 2013}

\newcommand{\by}{\mathbf{y}}
\newcommand{\bY}{\mathbf{Y}}
\newcommand{\R}{\mathbb{R}}

\newcommand{\M}{\mathcal{M}}

\newtheorem{Rem}{Remark}

\begin{document}


\title{Bayesian Analysis of ODE's: solver optimal accuracy and Bayes factors}

\author{
Marcos A. Capistr\'an$^{\rm a}$\thanks{Email: \textit{marcos@cimat.mx} \vspace{6pt}},
J. Andr\'es Christen$^{\rm a}$\thanks{Corresponding author. Email: \textit{jac@cimat.mx} \vspace{6pt}} and
Sophie Donnet$^{\rm b}$\thanks{Email: \textit{sophie.donnet@agroparistech.fr} \vspace{6pt}} \\
\vspace{6pt} $^{a}${\em Centro de Investigaci\'on en Matem\'aticas, Guanajuato, M\'exico};\\
\vspace{6pt} $^{b}${\em INRA, Unit\'e MIA, Equipe MORSE, France.}}


\maketitle

\begin{abstract}
In most relevant cases in the Bayesian analysis of ODE inverse problems, 
a numerical solver needs to be used.  Therefore, we cannot work with the exact 
theoretical posterior distribution but only with an approximate posterior deriving  from the error in
the numerical solver.  To compare a numerical and the theoretical posterior distributions we propose to use Bayes Factors (BF), considering both of them as models for the data at hand.  We prove that
the theoretical vs a numerical posterior BF tends to 1, in the same order
(of the step size used) as the numerical forward map solver does.  For higher order solvers
(eg. Runge-Kutta) the Bayes Factor is already nearly 1 for step sizes that would take far less
computational effort.  Considerable CPU time may be saved by using coarser solvers
that nevertheless produce practically error free posteriors.  Two examples are presented
where nearly 90\% CPU time is saved while all inference results are identical to using a solver with
a much finer time step.         
\end{abstract}

KEYWORD: Inverse Problems; Bayesian Inference; Bayes Factors; Numerical Analysis of ODE's.


\section{Introduction}\label{sec:intro}

\subsection{Context and issues}

In a comprehensive review of recent publications on the Bayesian Analysis of Inverse Problems
it is clear that there is a consistent growth of interest in the uncertainty quantification approach provided
by the Bayesian paradigm.  However, we also notice that some of the potential strength of the Bayesian approach 
is currently underexploited, namely, 1) prediction, 2) model selection and (let alone) 3) decision making
under uncertainty.  Recent reviews on the Bayesian Analysis of inverse problems
\citep{MOHAMMAD-D2006, KAIPIO&FOX2011, WATZENIG&FOX2009, KAIPIO&FOX2011, WOODBURY2011}
do not discuss these former points in detail.  We have now access to important
theoretical results on the definition of posterior distributions in infinite dimensions and regularity
conditions for correct approximate schemes through numerical, finite dimension posteriors
\citep[eg.][]{SCHWAB&STUART2012} that provide a sound theoretical
background to the field.  There are several applications
in the (Bayesian dominated) field of image processing of various kinds
\citep[][to mention some recent references]{ZHUetAl2011, CAIetAl2011, FALLetAl2011, CHAMAetAl2012, KOLEHMAINENetAl2007, NISSINENetAl2011, KOZAWAetAl2012}, and also we find a whole range of
emerging application areas in the Bayesian Analysis of inverse problems
\citep{CALVETTIetAl2006b, KEATSetAl2010, CUIetAl2011, WAN&ZABARAS2011, HAZELTON2010, KAIPIO&FOX2011}.
However, only a handful of publications mentions or uses Bayesian predicting tools, that is, the posterior (predictive)
distribution of yet to observe variables
\citep{VEHTARI&LAMPINEN2000, SOMERSALOetAl2003, KAIPIO&FOX2011, CAPISTRANetAl2012}, and even less consider formally the model selection and model comparison tools
developed in Bayesian statistics.  We intend to contribute to the formal and more systematic use
of the latter in the context of Inverse Problems. 

\medskip
 Predictive power is always a desirable property of mathematical models and inference, 
beyond parameter estimation, for model parameters that may or may not have straightforward physical meaning.
We believe that comparing forward models as \textit{statistical} models is the way to proceed when predictive power is of main
interest.  The Bayesian model comparison and model averaging tools, in particular pairwise model comparison using
Bayes Factors, is in such case the main tool to be used in this context, as far as predictive power is concerned \citep{HOETINGetAl1999}.
In particular, this idea should be used when analyzing the
numerical vs. the theoretical versions of the resulting posterior distribution.  That is,
the forward map, defined as the solution of a system of ODE's (or PDE's), represents a
complex regressor that is only theoretically defined.  The actual usage of the model necessarily
involves a numerical solver that includes an approximation error depending on the solver step size $h$.  Therefore,
there is a theoretical posterior distribution $P_{\Theta | Y}( \theta | y)$ and the approximate
$P_{\Theta | Y}^h( \theta | y)$ posterior.
Recently a series of papers \citep[eg.][]{SCHWAB&STUART2012}
discuss regularity conditions under which the latter
tends to the former as the approximation error tends to zero, using a suitable metric.
A metric comparison (ie. $|| P_{\Theta | Y}( \cdot | y) -  P_{\Theta | Y}^h( \cdot | y) ||$) is
useful to proving the required convergence theorems, but more practical considerations
will be needed when evaluating the relative benefits of a numerical approach with
a particular solver step size $h$ (for data $y$).

\medskip
 We believe that both
$P_{\Theta | Y}( \cdot | y)$ and $P_{\Theta | Y}^h( \cdot | y)$ (and $P_{\Theta | Y}^{h_i}( \cdot | y)$
for any other $h_i$) should be compared as \textit{models}, being $P_{\Theta | Y}( \cdot | y)$
the reference but computationally very expensive or even only theoretically available model
and the approximate, for various solver precisions $h_1, h_2, \ldots$, as alternative,
less computationally demanding computing models.  Bayes factors may then be used to establish
a sound comparison, to balance predictive power on the one hand vs. solver CPU time on the other,
to establish a useful solver precision (that may even be far less precise than common usage but
achieve comparable results relative to the observations errors, the model characteristics etc. for
the problem at hand).  One difficulty here is that in most real applications the theoretical model
is unavailable.  We therefore establish how to approximate the Bayes factors, without having
the theoretical reference model, using solely the numerical solver approximation rates.

\subsection{Notation}  

Assume that we observe a process $\by = (y_1,\dots,y_n)$ at the discrete times
$t_1, \dots, t_n \in [0,T[ ^n$ such that
\begin{equation}\label{eqn.obsmodel}
y_i  = f(X_{\theta}(t_i)) + \varepsilon_i ,~~ \varepsilon_i \sim_{i.i.d.} \mathcal{N}(0,\sigma^2) \quad \quad \quad \quad  (\M)
\end{equation}
where $X_{\theta}$ is the solution of the following system of ordinary differential equations, namely the  Forward Model,
\begin{equation}\label{eqn.odemodel}
\frac{dX_{\theta}}{dt} = F( X_{\theta}, t, \theta ); ~~ X_{\theta}(t_0) = X_0. 
\end{equation}
$\theta \in \Theta \subset \R^d$ is a vector of unknown parameters.
$F: \R^p \times [0, T[ \times \Theta \mapsto \R^p$ is a known function,
whose regularity properties  ensure the existence of a unique solution of the 
initial value problem  (\ref{eqn.odemodel}) (the regularity conditions of $F$ will be discussed in Section~\ref{sec:solvers}). 

\medskip
 $f : \R^p \rightarrow \R^k$ in (\ref{eqn.obsmodel}) is the observation function.
Many types of observation functions $f$ can be considered, modeling for instance the observation of a single component of the $p$-vector $X_{\theta}(t)$ or a (linear) combination of the components.  In this paper, for the sake of simplicity, we consider  a one dimensional observations
problem only, that is $k=1$.
Generalizations of our results to multivariate observations are possible and will be briefly mentioned in the Section~\ref{sec:discussion}. 

\medskip
 Any statistical decision  from the data $\by$ --such as  estimation, prediction or model
selection-- relies on the likelihood function (in a Bayesian and other paradigms)
\begin{equation}\label{eqn.lik}
P_{ \bY | \phi} (\by | \theta,\sigma) = \sigma^{-n}  (2\pi)^{-n/2} \exp\left\{-\frac{1}{2 \sigma^2} \sum_{i=1}^n (y_i -f(X_{\theta}(t_i)))^2  \right\}
\end{equation}
whose  expression involves the computation of  $X_{\theta}$, a solution of (\ref{eqn.odemodel}).
However, except in very simple cases, an explicit expression of the solution is in general not available (although its existence
is ensured by regularity assumptions on
$F$).  As a consequence, in practice, the system (\ref{eqn.odemodel})  is solved using a numerical solver and the inference is
performed, not on the previous ``exact" model but on an approximate model, namely
\begin{equation}\label{eqn.obsmodel.approx}
y_i  = f(X^h_{\theta}(t_i)) + \varepsilon_i ,~~ \varepsilon_i \sim_{i.i.d.} \mathcal{N}(0,\sigma^2)\quad \quad \quad \quad  (\M_h)
\end{equation}\label{eqn.lik.approx}
where $X_{\theta}^h$ denotes the approximate solution of (\ref{eqn.odemodel}) supplied by the numerical solver ($h$ being a precision parameter of the solver).
The new likelihood derived from model $\M_h$ is thus
\begin{equation}
P^h_{ \bY | \phi} (\by | \theta,\sigma) = \sigma^{-n}  (2\pi)^{-n/2} \exp\left\{-\frac{1}{2 \sigma^2} \sum_{i=1}^n (y_i -f(X^h_{\theta}(t_i)))^2  \right\} .
\end{equation}

\medskip
 Changing from the original statistical model is obviously not without consequences.
Since  there is no alternative but to use the approximate model
above (\ref{eqn.obsmodel.approx}),  there exists a real need in understanding and controlling  related consequences. In a Bayesian paradigm, any decision is
based on the posterior distribution.  A first natural choice is to compare the posterior distributions calculated from models
$(\M)$ and $(\M_h)$.  Such a study has been proposed by, for example, \citet{DONNETetAL2007}. 
However, when comparing models in a Bayesian context,
the natural tools are Bayes Factors. In this work, we recall their importance, propose an efficient way to compute them in this context
and study some theoretical aspects of their calculation when the exact model is not available. 

\medskip
 The paper is organized as follows. In Section \ref{sec:solvers} we discuss the choice of a solver and its required properties from a Bayesian Inverse Problem point of view. Bayesian inference is presented in Section \ref{sec:posterior} and some control results on the posterior distributions are cited in  Subsection \ref{subsec:control}.  Bayes factors are introduced and our main theoretical results are developed in Section  \ref{sec:bayes factor} . 
Our results are illustrated both with a simulation study and with real data in Section \ref{sec:numerical examples}.
Finally, a discussion of the paper is presented in Section \ref{sec:discussion}.

\section{ODE solvers from a Bayesian Inverse problem point of view}\label{sec:solvers}

Bayesian analysis for inverse problems strongly relies on the numerical approximation of the underlying ODE system. One can choose to use a standard (more or less advanced) implemented solver as a black-box in a sense assuming that no approximation is made on the model. But in our approach, we aim at understanding the influence of this approximation. As a consequence, we are interested in the inherent properties of the numerical solver. When it comes to qualifing a numerical solver, three properties arise, namely its error (local or global), its stability and its stiffness. The three of them are addressed in the following section although global error is surely  the most important one
as far as the purposes of the paper are concern. 

\medskip
 Only a restricted number of nonlinear ODEs has a solution in closed form.
Consequently, there is a plethora of numerical methods to solve the initial value 
problem~(\ref{eqn.odemodel}). Noteworthy are time-stepping methods based on Taylor 
approximation of the function $F$, multi-step methods and Runge-Kutta methods. 
These methods span through many orders of local accuracy. Beside the order of 
accuracy, standard requirements for a numerical method are consistency, convergence, and stability \citep{iserles1996first,quarteroni2007numerical}. However, even when these latter conditions hold, a common concern in the 
numerical solution of the initial value problem~(\ref{eqn.odemodel}) is error control. 

\medskip
 Two types of error can be considered, namely the local and the global errors.
In the following, we use the theory developed in \citet{quarteroni2007numerical}
to discuss the relationship between local and global errors for one-step Euler and 
fourth order explicit Runge-Kutta methods. 

\medskip
 Let  $h$ be the step size of the method.  We define a time grid as $ t_{n+1} = t_n + h$, for some fixed $h >0$.
Let $X_{\theta,n}$ be the solver approximation of $X_{\theta} (t_n)$.
One-step Euler and fourth order explicit Runge-Kutta methods have the form
\begin{equation}
\label{eq:euler}
X_{\theta,n+1} = X_{\theta,n} + hK(t_{n},X_{\theta,n},h,F),
\end{equation}
where 
\begin{equation}
K(t_{n},X_{\theta,n},h,F) = 
\left\{
\begin{array}{ll}
K_{1} & \mbox{for Euler}\\
\frac{1}{6}(K_{1}+2K_{2}+2K_{3}+K_{4}) & \mbox{for 4th ord. Runge-Kutta,}
\end{array}\right.
\end{equation}
with  
$$\begin{array}{cclccl}
K_{1}&=&F( X_{\theta,n}, t_{n}, \theta) & K_{2}&=&F( X_{\theta} + \frac{1}{2}hK_{1}, t_{n}+\frac{1}{2}h, \theta)\\
K_{3}&=&F( X_{\theta} + \frac{1}{2}hK_{2}, t_{n}+\frac{1}{2}h, \theta)&K_{4}&=&F( X_{\theta} + hK_{1}, t_{n} + h, \theta) .
\end{array}
$$
(There is also a 2nd order Rungue-Kutta, which will be used in Section~\ref{sec:exa1}, but we avoid presenting its details.)

\medskip
 Local truncation error $e_{n}$ is the error made in one step of the numerical method, that is
$$
e_n = e_{h}(t_n,\theta)=||X_{\theta}(t_{n})- X_{\theta}(t_{n-1}) - hK(t_{n},X_{\theta}(t_{n-1}),h,F)||_{2} .
$$
Global error $E_{n}$ is the difference between the computed solution and the 
true solution at any given value of  $t$ belonging to the grid
$$
E_n = E_{h}(t_n,\theta)=||X_{\theta}(t_{n})-X_{\theta,n}||_{2} .
$$
For (order one) Euler and fourth order Runge-Kutta methods $e_{n}$ is $O(h^{2})$ and $O(h^{5})$ 
respectively (we avoid the details for order two Rungue-Kutta, for which indeed $e_n$ is $O(h^{3})$).
Proposition~\ref{prop:error} establishes a relationship between global 
and local error for one-step numerical methods like Euler and Runge-Kutta.

\begin{assum}\label{assum}
We assume throughout the paper that the function $F$ in the right hand side 
of the initial value 
problem~(\ref{eqn.odemodel}) is continuously differentiable with respect 
to $X_{\theta}$ on a parallelepiped $\mathcal{R}= \{:t_{0}\leq t\leq t_{0}+a$, $||X_{\theta}-X_{0}||_{2}\leq b\}$ for each $\theta$. Hence a unique solution of the initial value problem exists in a neighborhood of 
$(t_{0},X_{0})$ for each $\theta$. $F$ as above implies $K$ is Lipschitz continuous 
on $X_{\theta}$, i.e. there is $L\in\mathbb{R}^{+*}$ such that if $(t,x)$, $(t,y)\in R$, then
$$
||K(t,x,h,F)-K(t,y,h,F)||_{2}\leq L||x-y||_{2} .
$$
\end{assum}

\begin{prop}\label{prop:error}
Assume that \emph{Assumption~\ref{assum}} holds  and for the solver at hand and
the local truncation error $e_{n}$ is $O(h^{p+1})$.  Then the global truncation error is $O(h^p)$.
\end{prop}

\begin{proof}
We have
$$
E_{n} =  ||X_{\theta}(t_{n}) - X_{\theta,n}||_{2} = ||X_{\theta}(t_{n}) - X_{\theta,n-1} -  h K(t_{n},X_{\theta,n-1},h,F)||_{2} .
$$
Adding and susbtracting the term $X_{\theta}(t_{n-1}) + h K( t_{n}, X_{\theta}(t_{n-1}), h, F)$ we obtain
\begin{align*}
E_{n} &\leq  e_{n} +  || X_{\theta}(t_{n-1}) - X_{\theta,n-1}  ||_{2}  +  h || K( t_{n}, X_{\theta}(t_{n-1}), h, F) - K(t_{n},X_{\theta,n-1},h,F) ||_{2} \\
	  &\leq e_{n} + E_{n-1} + h L || X_{\theta}(t_{n-1}) - X_{\theta,n-1}  ||_{2} \\
	  &\leq e_{n} + E_{n-1}(1 + hL).
\end{align*}	  
Consequently
\begin{align*}
E_{n} &\leq e_{n} + e_{n-1}(1 + hL) + \ldots + e_{1}(1 + hL)^{n-1}\\
           &\leq H h^{p+1} \frac{1-(1 + hL)^n}{1-(1 + hL)} = \frac{H}{L} h^p \{ ( 1 + hL)^n - 1 \} .
\end{align*}	
Therefore, since $h = \frac{l}{n}$ ($l=t_n-t_1$),
$$
E_{n} \leq \frac{H}{L} h^p \left\{ \left( 1 + \frac{l L}{n} \right)^n - 1 \right\} \leq \frac{H}{L} e^{l  L}  h^p.
$$
\end{proof}

From Proposition \ref{prop:error}, we deduce that, for  any explicit one-step method of order $p$ such as Euler ($p=1$)
and Runge-Kutta  ($p=2$ or $p=4$) schemes, the global error is of order $O(h^p)$
for $h$ small enough.  If we set $X_{\theta}^h(t_n) = X_{\theta,n}$, note that we have proved that
$$
\max_{t \in \{t_1, t_2, \ldots , t_n\}} || X_{\theta}(t) - X_{\theta}^h(t) || \leq C_{\theta} h^p
$$
since $H$ and $L$ above depend on $\theta$.  This global error order control will be needed in Sections~\ref{subsec:control}
and~\ref{sec:comparing} to prove our main result.

\medskip
 In the numerical community,  the control of the error  means keeping the error $e_n$ or $E_n$ under a fixed level
(beyond the above mentioned asimptotic results), along the application of the scheme. 
The local error $e_n$ can be controlled directly, using for instance,  the ``Milne device'' \citep{iserles1996first}.
However, there are not known general methods to control global 
error $E_n$, although some methods exist to estimate it,  for instance, those 
relying on adjoint state analysis, see~\cite{cao2004posteriori,lang2007global}.  In the results that follow we
do not require an estimation of this global (or local) error, solely the knowledge of the global error order for the
solver at hand.

\medskip
 Another important issue in the numerical solution of problem~(\ref{eqn.odemodel}) is stiffness.
According to~\cite{lambert1991numerical}, if a numerical method with finite region
of absolute stability applied to system~(\ref{eqn.odemodel}) is forced to use in a 
certain interval in $t$ a steplength excessively small in relation to the 
smoothness of the exact solution in that interval, then system~(\ref{eqn.odemodel}) 
is said to be stiff in that interval. We remark that many systems of ODE modeling 
real life phenomena are stiff, see~\cite{gutenkunst2007universally}.

\medskip
 Software implementing time-stepping methods mentioned above must address
error control, stability and stiffness issues. Actually, local error control 
mechanisms can cope with stiffness at the expense of taking very small step-sizes.
Therefore, advanced features like variable order methods, and variable stepsize
methods have been developed and implemented in libraries of common high level
programing languages like R, Matlab and Python-Scipy. 

\medskip
 The results shown in this paper assume a fixed step method. 
We only work with the Euler and Rugue-Kutta methods (orders $p=1,2,4$ respectively).

\section{Bayesian inference for inverse problems: practical and theoretical aspects}\label{sec:posterior}

There are some excellent reviews concerning the Bayesian analysis of Inverse Problems
\citep{ KAIPIO&SOMERSALO2005, FOXetAl1999} and for a more detailed description these or other sources should be consulted.
In this section we only present the particular aspects of the field relevant to the development of our results.

\medskip
 Any  Bayesian statistical decision (such as estimation) is  based on the posterior distribution,
given by the Bayes formula 
\begin{equation}\label{eqn.post}
P^h_{ \Phi | \bY }( \theta,\sigma | \by ) =  \frac{P^h_{ \bY | \Phi }( \by |  \theta,\sigma) P_{\Phi}(\theta,\sigma)}{P^h_{\bY} (\by)}
\end{equation}
where $P_{\Phi}(\theta,\sigma)$ is the prior distribution on $(\theta,\sigma)$   and
$$
P^h_{\bY} (\by) = \int P^h_{ \bY | \Phi }( \by |  \theta,\sigma) P_{\Phi}(\theta,\sigma) d\theta d\sigma
$$
is the normalization constant,  also called the \emph{marginal likelihood} of data $\by$.

\medskip
 An estimator  of $(\theta,\sigma)$ is derived from some characteristics of the posterior distribution. Various approaches can be adopted.  
Either the estimator is taken as the mode of the posterior distribution, resulting into the MAP estimator
$(\hat{\theta}^{MAP},\hat{\sigma}^{MAP}) = \arg \max _{(\theta,\sigma) }P^h_{ \Phi| \bY }( \theta,\sigma | \by )$ 
or  the estimator derives from the minimization of a loss function $L(\theta,\sigma| \eta)$, that is
$ (\hat{\theta},\hat{\sigma})  = \arg \min _{(\theta,\sigma)} \int L(\theta,\sigma| \eta)P^h_{ \Phi | \bY }(\theta,\sigma | \by ) d\theta d\sigma $.
If the loss function $ L(\theta,\sigma| \eta)$ is equal to the $L^2$-distance, then the corresponding estimator is the posterior mean
$ (\hat{\theta}^{L^2},\hat{\sigma}^{L^2}) = \int (\theta,\sigma ) P^h_{\Phi| \bY }( \theta,\sigma  | \by ) d\theta d\sigma $.
 If the loss function $ L(\theta| \eta)$  is equal to the $L^1$-distance, then  the corresponding estimator is the posterior median, that is
 $ (\hat{\theta}^{L^1},\hat{\sigma}^{L^1}) =  (F^h)^{-1}_{ \Phi| \bY }\left( \frac{1}{2}\right)$, etc.  In general, posterior expectations of all types
 (ie. $\int g(\phi) P^h_{\Phi| \bY }( \phi  | \by ) d\phi$, for some measurable function $g$) are used to explore
 the posterior distribution, and these include all marginal distributions for individual parameters or posterior probabilities of
 specific regions of interest.

\medskip
 Besides some few simple conjugate models,  the normalizing constant $P^h_{\bY} (\by)$ has no explicit analytic expression and
therefore the above estimators can not be directly calculated. If the dimension
of $\Phi$ is 1 or 2 we could rely on numerical integration to obtain the normalizing constant $P_{\bY}^h (\by)$. In larger dimensions, the standard solution is to resort  to Monte Carlo methods to approximate the estimators.
Let  $(\theta^{(l)}, \sigma^{(l)})_{i=1\dots L}$ be a sample  from the posterior distribution, the $L^2$ estimator for instance, is approximated as
$(\hat{\theta}^{L^2},\hat{\sigma}^{L^2})  = \left( \frac{1}{L} \sum_{l=1}^M \theta^{(l)},   \frac{1}{L} \sum_{l=1}^L \sigma^{(l)}\right)$.

\medskip
 Simulation from the posterior distribution is not a direct task and Markov Chain Monte Carlo algorithms
\citep[see][for a didactic review]{Robert&Casella2004} are standard tools to sample from the posterior distribution $P^h_{ \Phi| \bY }( \theta,\sigma | \by )$. However,  such algorithms have to be carefully designed when the evaluation of the regression function of the model (and consequently of the likelihood function) is computationally intensive.  In the next section, we recall the basics of MCMC algorithms and discuss their application to inverse problems and the calculation of Bayes' Factors.

\subsection{Sampling from  the posterior distribution in Inverse problems}\label{sec:MCMC}

Markov Chain Monte Carlo (MCMC) is specially suited
for sampling from complex multidimensional distributions and is ubiquitous in modern Bayesian analyses.
We do not intend to present a comprehensive review of MCMC here but merely state the basic principles to consider some aspects
of implementing MCMC in the Inverse Problem context; otherwise the reader is referred to \citet{Robert&Casella2004}.
The principle of MCMC algorithms is to generate a Markov chain whose invariant distribution is the distribution of interest
$P_{ \Phi | \bY }^h( \theta,\sigma | \by )$. Many versions have been proposed in the literature. Among those, the Gibbs algorithm
and the Metropolis-Hasting algorithms are the most used and quickly described below.

\medskip
 The Gibbs sampler is a very popular MCMC algorithm.  However, only in canonical cases
 (eg. conditional conjugacy) it make sense to be used.   In our Inverse Problem settings,  this is
not the case and the general Metropolis-Hastings MCMC algorithm needs to be used instead.

\medskip
 The Metropolis-Hastings (MH) algorithm starts by defining
a proposal (or instrumental) conditional distribution $q( \phi' | \phi )$, which we are able to simulate from for any
$\phi$ in the support of the posterior distribution.  At iteration $(\ell)$:  
\begin{enumerate}
\item a proposed value $\phi'$ is simulated given the current state of the Markov
Chain $\phi^{(\ell)}$ from the proposal $q( \cdot | \phi^{(\ell)} )$;   
\item the proposed point $\phi'$ is accepted as the new
point in the Markov chain $\phi^{(\ell+1)}$ with  probability $ \rho( \phi^{(\ell)}, \phi' )$: 
$$\
phi^{(\ell +1)} = \left\{
\begin{array}{ccl}
\phi' & \mbox{with probability} & \rho( \phi^{(\ell)}, \phi' )\\
\phi'^{(\ell)} & \mbox{with probability} & 1-\rho( \phi^{(\ell)}, \phi' )
\end{array}
 \right.
 $$
where the acceptance probability $ \rho( \phi^{(\ell)}, \phi' )$ is equal to 
\begin{eqnarray*}
\rho( \phi^{(\ell)}, \phi' ) &=& \min\left\{1, \frac{P_{ \Phi | \bY }^h( \phi' | \by )}{P_{ \Phi | \bY }^h( \phi^{(\ell)} | \by )}
\frac{q( \phi^{(\ell)} | \phi' )}{q( \phi' | \phi^{(\ell)} )}\right\}\\
& =&  \min\left\{1, \frac{P_{\bY |\Phi}^h(  \by |  \phi'  )P_{\Phi}(\phi')}{P_{\bY |\Phi}^h(  \by |  \phi^{(\ell)}  )P_{\Phi}( \phi^{(\ell)})}
\frac{q( \phi^{(\ell)} | \phi' )}{q( \phi' | \phi^{(\ell)} )}\right\} .
\end{eqnarray*}
\end{enumerate}
Commonly a series of proposal distributions $q_1, q_2, \ldots , q_m$ are
entertained, leading to a series of $m$ transition kernels that are systematically or randomly scanned (the latter
leads to the desirable reversibility property) to form an easily provable convergent chain to the
posterior $P_{ \Phi | \bY }^h( \cdot | \by )$.

\medskip
 In any case, at each iteration of the MH MCMC we need to evaluate
the (unnormalized) posterior (It is therefore crucial to minimize the number of iterations in the MCMC and
the number of likelihood evaluations).

\medskip
 Optimizing MCMC algorithms (that is to say minimizing the number of iterations)  has been a very active research topic in the last decade.  There are
adaptive algorithms \citep{HaarioEtAl1998, Atchade&Rosenthal2005} that try to learn from previous steps
of the chain to adapt the proposals $q_1, \ldots , q_m$.  These methods require additional regularity conditions
on the adaptive scheme, model and prior that might limit their applicability.
\citet{Christen&Fox2010} also propose the  t-walk, which self adjusts keeping two points in the parameter
space, and that commonly samples with reasonable efficiency.  However, robust, multipurpose, automatic and
optimal methods are still far away in the MCMC horizon (to make an optimistic metaphor).

\medskip
 When it comes out to inverse problem specifically, various savings can be considered. 

\medskip
 $\bullet$ First, a naive but straightforward computational economy is to  save
 $U^{(\ell)} = -\log P_{ \bY | \Phi }^h( \by |  \phi^{(\ell)} ) - \log P_{\Phi}( \phi^{(\ell)} )$.
 Indeed, this quantity will be used for any new simulated value $\phi'$ until the chain reaches a new point $\phi^{(\ell+1)}$.
 We will see in Section~\ref{subset:comp Bayes factor} that these quantities can also be recycled for model comparison purposes. Note that these $U^{(\ell)}$ are very useful for convergence analysis.
 
\medskip
 $\bullet$ Second, in this Inverse Problem setting we may exploit the possibility that a coarser numerical solver with 
a higher error rate (or equivalently a larger step size $h_1$)  might lead to a far less CPU demanding approximate
calculation of the likelihood $P_{ \bY | \phi }^{h_1}( \by |  \phi^{(i)} )$.  More precisely, \cite{Christen&Fox2005} suggest to propose a candidate $\phi'$ and test its promising acceptance  probability using an approximate and cheap calculation  of the likelihood (using the coarser numerical solver). If this $\phi'$ has a ``good" probability to being accepted, the full blow and expensive calculation of the likelihood is used to firmly  accept or reject $\phi'$. This two-step or ``delay acceptance MH'' algorithm  may save substantial CPU time during the MCMC \citep[see][]{Christen&Fox2005}.

\subsection{Theoretical results on the error control of the approximate posterior distribution}\label{subsec:control}

After having examined how the standard algorithmic tools for Bayesian inference can be adapted or designed specifically to the Inverse problem context, we may want to understand and control the consequences implied by the use of an approximate model from a theoretical point of view. Such a  result can be found in \citet{DONNETetAL2007} who compare  the posterior distributions of the exact and approximate models respectively --namely  $P_{ \Phi | \bY  }$ and  $P^h_{ \Phi | \bY  }$--  through the total variation distance.

\begin{prop}
Assume that $\phi = (\theta, \sigma)$ remains in a compact set and that the 
numerical scheme of step size $h$ is such that $\{t_1,\dots,t_n\}\subset h \mathbb{N}$ and
\begin{equation}\label{eqn:globerr}
\max_{t \in \{t_1,\dots,t_n\}} \|X_{\theta}(t)-X^h_{\theta}(t)\|_{\R^p} \leq C_{\theta} h^p.
\end{equation}
Also assume that  the observation function $f$ is differentiable with a bounded derivative. Then 
there exists a constant $C_{\by}$ such that for every $(\theta, \sigma)$ and $h$ small enough 
\begin{equation}\label{eq:ineq1}
|P_{ \Phi | \bY  }(\theta,\sigma;\by)-  P^h_{ \Phi | \bY  }(\theta,\sigma;\by)| \leq  C_{\by}  \pi(\theta,\sigma)  h^p .
\end{equation}
As a consequence, 
\begin{equation}\label{eq:ineq2}
D_{T.V}(P_{ \Phi | \bY  },  P^h_{ \Phi | \bY  }) \leq  C_{\by}  h^p
\end{equation}
where $D_{T.V}$ is the total variation distance. 
Moreover, \begin{equation}\label{eq:ineq3}
\| (\hat{\theta}^{L^2},\hat{\sigma}^{L^2})  - (\hat{\theta}^{h,L^2},\hat{\sigma}^{h,L^2}) \|  \leq  h^p C'_{\by} 
\end{equation}
\end{prop}

\begin{proof}
\cite{DONNETetAL2007}'s results are developed for mixed effects models in a maximum likelihood context but can be adapted to models~(\ref{eqn.obsmodel}) and ~(\ref{eqn.obsmodel.approx}). Inequality (\ref{eq:ineq1}) is derived from Theorem 4 of  \cite{DONNETetAL2007}. The control on the total variation distance is derived directly. The inequality (\ref{eq:ineq3}) is obtained as follows
$$
\begin{array}{rcl}
&& \| (\hat{\theta}^{L^2},\hat{\sigma}^{L^2})  -  (\hat{\theta}^{h,L^2},\hat{\sigma}^{h,L^2}) \| \\
&&= \|\int (\theta,\sigma ) P_{\Phi| \bY }( \theta,\sigma  | \by ) d\theta  -  \int (\theta,\sigma ) P^h_{\Phi| \bY }( \theta,\sigma  | \by ) d\theta d\sigma\| \\
 &&= \|\int (\theta,\sigma ) \left[P_{\Phi| \bY }( \theta,\sigma  | \by ) - P^h_{\Phi| \bY }( \theta,\sigma  | \by ) \right]d\theta d\sigma\|\\
&&\leq \int \|(\theta;\sigma)\|   C_{\by}  P_{\Phi}(\theta,\sigma)  h^p d\theta d\sigma\\
&&\leq     h^p C_{\by} \int \|(\theta;\sigma)\|    P_{\Phi}(\theta,\sigma)  d\theta d\sigma 
\end{array}
$$
\end{proof}

Note that obtaining the same type of control on the MAP estimator requires additional regularity assumptions on the posterior distributions,
such as unicity of the extremum and properties on the second derivatives. This leads to encourage the use of $L^2$ estimates over the  MAP estimates in this context.

\medskip
 \cite{COTTERetAl2010} also proposes control results of the same type. Even if these results  provide  interesting theoretical error controls, they rely on unknown constants and so can not be used as such in practice.
In the next section, we adopt a Bayes factor point of view and highlight that such an approach leads to results of more practical interest.

\section{Model selection and Bayes Factor for inverse problems}\label{sec:bayes factor}

In Bayesian inference, model selection is performed using the Bayes factors whose principle is recalled here in a general context. Let $\by$ be the observations and let $\M_1$ and $\M_2$ be two models  in competition. Each model $\M_i$  is defined through a likelihood depending on a set of parameters and a prior distribution on the parameters. More precisely, 
$$
(\M_1) \left\{
\begin{array}{ccl}
\by &\sim& P^1_{\bY | \Phi_1}(\by|\phi_1)\\
\phi_1 & \sim & P_{\Phi_1}^1(\phi_1) 
\end{array}\right. 
\quad  \quad 
(\M_2) \left\{ \begin{array}{ccl}
\by &\sim& P^2_{\bY | \Phi_2}(\by|\phi_2)\\
\phi_2 & \sim & P_{\Phi_2}^2(\phi_2) 
\end{array}\right. .
$$ 
Consider a prior distribution on the set of the models $\{\M_1, \M_2\}$,  the decision between the competing
models $\M_1$ and $\M_2$ is based on the ratio of their respective posterior probabilities
$$
\frac{P( \M_2 | \by )}{P( \M_1 | \by )} = \frac{P_{\bY}^2(\by)}{P_{\bY}^1(\by)} \frac{ P(\M_2)}{ P(\M_1)}
$$
where $P_{\bY}^i(\by)$ is the `integrated likelihood' or the marginal distribution of $\bY$ of model $\M_i$, namely
$$
P_{\bY}^i(\by) = \int P^i_{\bY | \Phi_i}(\by|\phi_i) P_{\Phi_i}(\phi_i)d\phi_i .
$$
Finally, the comparison of models relies on the computation of the marginal likelihoods which has  been the object of a rich literature.  Two approaches may be cited.  One consist in running a specific MCMC to approximate the quantity of interest \citep[see][for a review]{HanetCarlin2001}. The second relies on Monte Carlo approximations of the marginal likelihood $P_{\bY}^i(\by)$  \citep[see for instance][]{ChenAndShao1997}.

\subsection{Computation of Bayes Factors in an Inverse Problem context}\label{subset:comp Bayes factor}

 In our inverse problem context --where each iteration of the MCMC requires the computationally intensive approximation of an ODE -- we would like to avoid increasing the computational burden by using a specific MCMC and  would prefer recycling the outputs of the MCMC algorithm into a Monte Carlo strategy. An answer can be found in the Gelfand and Dey's estimator \citep{GelfandetDey1994}.

\medskip
 More precisely, a standard solution  to compute the marginal likelihood is to propose an estimation based on a Monte Carlo approximation.  Let $\{\phi_i^{(l)}\}_{l=1\dots L}$ be an i.i.d. sample
from a  proposal distribution $\pi_{IS}$ then the following estimator
$$
\widehat{p}_i = \frac{1}{L} \sum_{l=1}^L P^i_{\bY | \Phi^{(l)}_i}(\by|\phi^{(l)}_i) \frac{P_{\Phi_i}(\phi^{(l)}_i)}{\pi_{IS}(\phi^{(l)}_i)}
$$
supplies a convergent and non biased estimator of $P_{\bY}^i(\by) $. 
However, in the inverse problems context (when one or both models $\M_i$ are defined through ODE without explicit solution)  such a strategy requires the evaluation of the likelihood and thus the evaluation of an approximate
solution of the dynamical system for each newly generated value of parameters $\phi^{(l)}_i$, which can be burdensome from a computational time point
of view. 

\medskip
 The Gelfand and Dey's  estimator is an alternative solution to this situation. 
Assume that (as it is in standard situations),  $P_{ \Phi | \bY }^h( \theta, \sigma  | \by )$ is sampled using an intensive
Monte Carlo procedure, typically a Metropolis-Hasting MCMC. (Note that for ease of presentation we avoid the exponent  $i$ or $h$ and, notationally, consider the posterior distribution $P_{ \Phi | \bY }( \theta, \sigma  | \by )$ and estimation of the normalization constant $P_\bY(\by)$) 

\medskip

Assume that the prior distribution $P_\Phi$ is absolutely continuos w.r.t the Lebesgue measure, and thus $P_\Phi( \theta, \sigma )$ and $P_{\Phi | \bY}( \theta, \sigma | \by)$
are densities in the usual sense.  The Gelfand and Dey's  estimator  relies on  the obvious following expression
$$
\left[ P_\bY(\by) \right]^{-1} = \int \frac{\alpha( \theta, \sigma)}{P_{\bY | \Phi}( \by |  \theta, \sigma) P_{\Phi}( \theta, \sigma)}
P_{ \Phi | \bY }( \theta, \sigma  | \by ) d\theta d\sigma ,
$$
where $\alpha$ is  any density ($\int \alpha( \theta, \sigma) d\theta d\sigma = 1$) with support containing the support of the
posterior.
  
Now, let  $\theta^{(1)}, \sigma^{(1)}, \theta^{(2)}, \sigma^{(2)}, \ldots ,\theta^{(L)}, \sigma^{(L)}$  be a MCMC sample of the
posterior   $P_{ \Phi | \bY }( \theta, \sigma  | \by )$. Considering the above expression, the desired marginal may be approximated by 
\begin{equation} \label{eqn.estMarg}
\widehat{P}_\bY(\by) =  \left[ \frac{1}{L} \sum_{l=1}^L
\frac{\alpha( \theta^{(l)}, \sigma^{(l)} )}{P_{\bY | \Phi}( \by |  \theta^{(l)}, \sigma^{(l)} ) P_{\Phi}( \theta^{(l)}, \sigma^{(l)} )} \right]^{-1} .
\end{equation}
The choice of $\alpha$ conditions the quality of the estimator (its variance). If $\alpha(\theta,\sigma) = \pi(\theta,\sigma)$,  the estimator $\widehat{P}_\bY(\by)$ is the harmonic mean which is known to have a dramatic  unstable behaviour (infinite variance) in some cases. 
Best strategies are those that use a weighting $\alpha$ density that  stabilizes this estimators, for instance using somehow an $\alpha$ that 
resembles $P_{\bY | \Phi}( \by |  \theta, \sigma) P_{\Phi}( \theta, \sigma)$.  A simple calculation leads to the fact that using a
\textit{thinner} tailed $\alpha$ (as oppose to the result in importance sampling) is better suited to obtaing a finite variance for
the estimator above.  

\medskip
 Moreover, in an inverse problem context, it is critical to avoid recalculating the likelihood  $P_{ \bY | \Phi }( \by |  \theta, \sigma) $ since it involves numerically solving the ODE system in (\ref{eqn.odemodel}). 
As a consequence we propose to  proceed as follows.

\medskip
 $\bullet$  At each iteration of a typical MH MCMC, the computation of the probability of acceptation requires to evaluate
$P_{ \Phi | \bY }( \by |  \theta^{(l)}, \sigma^{(l)} ) P_{\Phi}( \theta^{(l)}, \sigma^{(l)} )$. 
After the burn in period, we save these values, letting
$U_l = -\log P_{ \Phi | \bY }( \by |  \theta^{(l)}, \sigma^{(l)} ) - \log P_{\Phi}( \theta^{(l)}, \sigma^{(l)} )$.

\medskip
 $\bullet$ A small subsample of $ \theta^{(l)}, \sigma^{(l)} $, typically of size less than 1,000, is then used
to create a Kernel Density Estimate (KDE), which we will use as our weighting density $\alpha$.  This KDE is (typically)
a mixture of Gaussians, with support in the whole space, and will roughly resemble the posterior $P_{\Phi | \bY}$.  

\medskip
 $\bullet$ Let
$A_i = - \log \alpha( \theta^{(l)}, \sigma^{(l)} )$, then our estimate becomes
$$
P_\bY(\by) \approx \left[ \frac{1}{L} \sum_{l=1}^L \exp( U_l - A_l ) \right]^{-1} .
$$
This procedure is fast
and basically is a byproduct of the MCMC sample, with little CPU burden added.  There are robust and fast KDE's routines
available in popular programming languages like R and Python-Scipy. 

\begin{Rem} 
Note that $P_{\bY | \Phi}( \by |  \theta, \sigma) P_{\Phi}( \theta, \sigma)$ needs to be known exactly, and be coded accordingly which is not the case  in some situations, 
for example, when the prior is not normalized and only implicitly defined etc.  
\end{Rem}

\subsection{Comparing the exact and approximate  models through Bayes factors}\label{sec:comparing}

 In section \ref{subsec:control}, we compared the true and approximate models (models \ref{eqn.obsmodel} and \ref{eqn.obsmodel.approx}) through a derived quantity,
that is to say their corresponding posterior distribution or the estimators they supplied. However, a more natural way to
compare models is to use the Bayes factors.   Let $(\M)$ and $(\M_h)$ be the two following models
$$
(\M) \left\{
\begin{array}{ccl}
y_i  &=& f(X_{\theta}(t_i)) + \varepsilon_i \\
 \varepsilon_i &\sim&_{i.i.d.} \mathcal{N}(0,\sigma^2)\\
 (\theta, \sigma) &\sim&   P_{\Phi}(\theta,\sigma)
\end{array}\right. 
\quad  \quad 
(\M_h) \left\{
\begin{array}{ccl}
y_i  &=& f(X^h_{\theta}(t_i)) + \varepsilon_i \\
 \varepsilon_i &\sim&_{i.i.d.} \mathcal{N}(0,\sigma^2)\\
 (\theta, \sigma) &\sim&   P_{\Phi}(\theta,\sigma) .
\end{array}\right.
$$
Comparing $\M$ and $\M_h$ through a Bayes factor requires to compute the following quantity
$$
B_{\M,\M_h} = \frac{P( \M | \by )}{P( \M_h | \by )} = \frac{\int P_{\bY|\Phi}(\by | \theta,\sigma) P_{\Phi}(\theta,\sigma)d\theta d\sigma}{\int P^h_{\bY|\Phi}(\by | \theta,\sigma) P_{\Phi}(\theta,\sigma)d\theta d\sigma}\frac{P(\M) }{ P(\M_h)} .
$$
However, $\int P_{\bY|\Phi}(\by | \theta,\sigma) P_{\Phi}(\theta,\sigma)d\theta d\sigma$ is not known in general,
since it involves the theoretical model.
In order to understand the behavior of the Bayes Factor, we study $B_{\M,\M_h}$ for small $h$ and obtain the following result. 

\begin{prop}\label{prop:2}
Assume that the numerical solver is such
that the global error truncation can be written as 
$$
E_h(t,\theta) = X_{\theta}^h(t) -X_{\theta}(t) = O(h^p) ,
$$
where $h$ is the stepsize of the method. In addition, assume that the observation function $f$ is differentiable on $\{X_{\theta}(t), \theta \in \Theta, t\in [0,T]\}$. 
Then, there exists a constant $B(\by) \in \mathbb{R}$ (which does not depend on $h$) such that
 $$
 \frac{P_\bY(\by)}{P^h_\bY(\by)}  \simeq 1+  B(\by) h^p  .
 $$ 
\end{prop}

\begin{proof}
Using the asymptotic behavior of the global error truncation and assuming that $f$ is differentiable, we can write
\begin{equation}\label{eqn:Dh}
D_h(t,\theta) = f(X_{\theta}^h(t)) - f(X_{\theta}(t)) =  \nabla f \left( X_{\theta}(t)\right)(X_{\theta}^h(t) -X_{\theta}(t)) +  O(h^p) 
\end{equation}
This approximation allows us to obtain a development of the marginal likelihood.
Let $R_h(\phi) = \frac{P^h_{\bY | \Phi}(\by|\theta,\sigma)}{P_{\bY | \Phi}(\by|\theta,\sigma)}$, then
\begin{eqnarray*}
P^h_\bY(\by)  &=& \int P_{\bY | \Phi}^h( \by |  \phi) P_{\Phi}(\phi) d\phi =
\int P_{\bY | \Phi}(\by| \phi)R_h(\phi) P_{\Phi}(\phi) d\phi \\
&=& P_\bY(\by) +  \int P_{\bY | \Phi}(\by| \phi)(R_h(\phi)-1) P_{\Phi}(\phi) d\phi .
\end{eqnarray*}
We see that
\begin{eqnarray} \label{eqn.rtheta}
R_h(\phi)  -1   &=& \exp \left\{ -\frac{1}{2\sigma^2}\sum_{i=1}^n \left[f(X^h_{\theta}(t_i))-f(X_{\theta}(t_i))\right]^2 + \right.\nonumber \\
         && \left.  2\left[y_i-f(X_{\theta}(t_i))\right]\left[f(X_{\theta}(t_i))-f(X^h_{\theta}(t_i))\right]  \right\}  -1 \nonumber \\
         &=& - \frac{1}{2\sigma^2}\sum_{i=1}^n D_h(t_i,\theta) ^2 + 2 (y_i-f(X_{\theta}(t_i)))D_h(t_i,\theta) \nonumber \\
         && + O(D_h(t_i,\theta)^2) ,
\end{eqnarray}
since $e^x -1 = x + O(x^2)$ for $x$ small enough.
Using the expression in (\ref{eqn:Dh}) for $D_h$ and the solver global error control
assumption in (\ref{eqn:globerr}) we must have that $R_h(\phi)-1  = O(h^p)$.
From this we get
$$
P^h_\bY(\by) = P_\bY(\by) + O(h^p)
$$ 
implying that for $h$ small enough, 
$$
\frac{P_\bY(\by)}{P^h_\bY(\by)} \simeq 1+ B(\by)h^p .
$$ 
\end{proof}

\medskip

\begin{corollaire}\label{prop:2}
If $\hat{g} = \int g(\phi)  P_{ \Phi | \bY }( \phi | \by ) d\phi$ and $\hat{g}^h = \int g(\phi)  P^h_{ \Phi | \bY }( \phi | \by ) d\phi$
exists, then $| \hat{g}^h - \hat{g} | = \frac{P_\bY(\by)}{P^h_\bY(\by)} B_g(\by)h^p = O(h^p)$.
\end{corollaire}

\begin{proof}
Note that
$| \hat{g}^h - \hat{g} | = \left|
\int g(\phi) R_h(\phi) \frac{P_\bY(\by)}{P^h_\bY(\by)} P_{ \Phi | \bY }( \phi | \by ) d\phi  - \hat{g} \right|$ and therefore
$| \hat{g}^h - \hat{g} | = \frac{P_\bY(\by)}{P^h_\bY(\by)}  \left|
\int g(\phi) (R_h(\phi) -1) P_{ \Phi | \bY }( \phi | \by ) d\phi  - \left( \frac{P^h_\bY(\by)}{P_\bY(\by)} - 1 \right)  \hat{g} \right|$.
Combining (\ref{eqn.rtheta}) and the above theorem one reaches the result.
\end{proof}

Regarding the above result, we make the following comments and remarks.

\medskip
 $\bullet$ From (\ref{eqn.rtheta}), we note that the error in the regression term $D_h(t,\theta)$ is not important per se but with respect to the observation noise standard error $\sigma$. This obvious remark has consequences. It means that, when working on a statistical model involving the numerical approximation of a differential system, there is no need in choosing a step size the smaller as possible but adapting it such that the global error is small with respect to $\sigma$. This can allow computational time savings as illustrated on the following numerical examples. 

\medskip
 $\bullet$
  Note that $B( \by)$ only depends on the numerical method  and on the data, but not on the step size $h$. 	

\medskip
 $\bullet$
 The constant $B( \by)$ can be estimated. Indeed, an obvious method is to compute $P_\bY^{h_k}(\by)$ for various values of step size $\{h_k, k=1\dots K\}$. A simple linear regression of $P_\bY^{h_k}(\by)$ against $h_k^p$  gives  an estimation of $B(\by)$.  This means that using a multi-resolution computation of the $P_\bY^{h_k}(\by)$ on various approximate models, we are able to estimate the marginal likelihood of the true model. 

\medskip
 $\bullet$  This last point is of major importance. Assume that, two models $\M_1$ and $\M_2$ have to be compared, one or both of them being defined by a differential system without explicit solution. These two models can be compared per se. Indeed, whereas one could have thought that only the approximate models $\M_1^h$ and $\M_2^{h}$ where comparable, we establish that the both true or exact models $\M_1$ and $\M_2$ may be directly compared. 

\medskip
 In the next section,  we develop two examples where we estimate the Bayes factor of the theoretically exact model with approximate models.  As a consequence
the step size may be coarsened, obtaining basically the same posterior distributions, but at far lower CPU costs.

\section{Numerical examples}\label{sec:numerical examples}

\subsection{Logistic Growth Models}\label{sec:exa1}

We base our first numerical study on the logistic growth model which is a common model of population growth in ecology.
Recently it has also been used to model  tumors growth in medicine, among many other applications.
Let $X(t)$ be the size of the tumor to time $t$. The dynamics are governed by the following differential equation
\begin{equation}\label{eq:logistic}
 \frac{dX}{dt} =\lambda X(t) (K-X(t)), \quad X(0)=X_0
\end{equation}
with $r  = \lambda K$ being the growth rate and $K$ the carrying capacity e.g. $\lim_{t\rightarrow \infty} X(t) = K$.  
 Equation (\ref{eq:logistic}) has an explicit solution equal to
$$
X(t) = \frac{KX_0e^{\lambda Kt}}{K + X_0(e^{\lambda K t}-1)}.
$$
We simulate two synthetic data sets with the error model $ y_i  =X(t_i) + \varepsilon_i$, where $\varepsilon_i \sim \mathcal{N}(0,\sigma^2)$, and  the following parameters $X(0)=100, \quad  \lambda= 1, \quad K=1000, \quad \sigma = 1 \mbox{ or }  30.$
The datasets are plotted on Figure \ref{fig.data} for the two chosen values of $\sigma$.
We consider $26$ observations at times $t_i$  regularly spaced between $0$ and $10$. 
\begin{figure}
\centering
\begin{tabular}{c l}
\includegraphics[width=6cm]{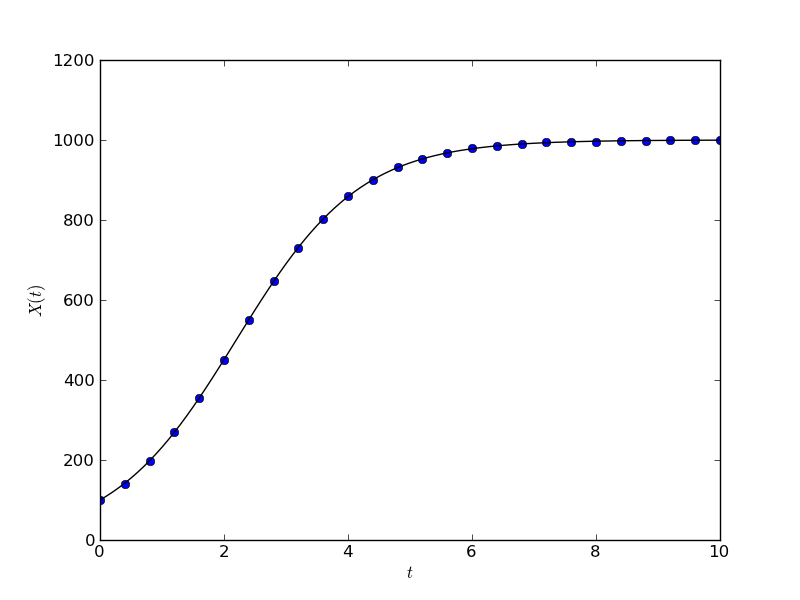} &
\includegraphics[width=6cm]{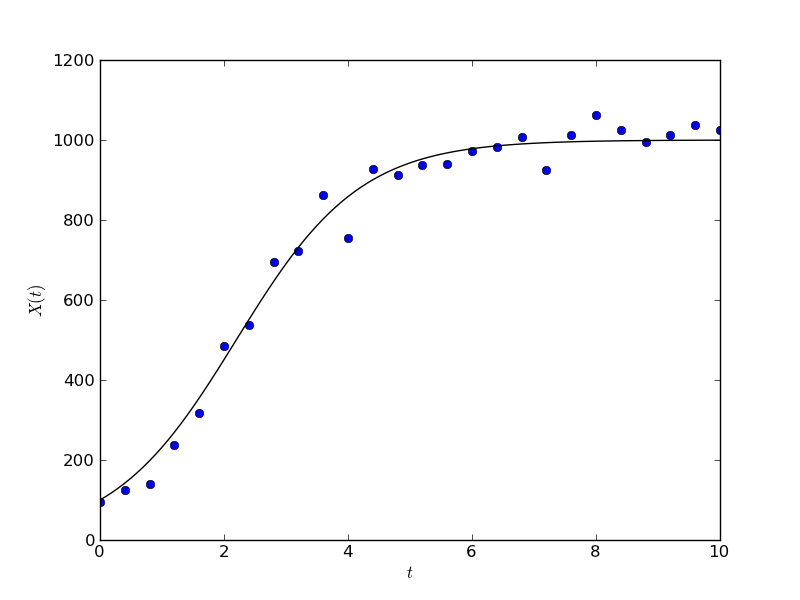} 
\end{tabular}
\caption{Synthetic data for the Logistic growth with $\lambda=1$,  $K=1000$ and  $\sigma=1$ (left) or $\sigma=30$ (right).}
\label{fig.data} 
\end{figure}

\medskip
 For this first toy example, $K$ is
taken as known and inference is concentrated on the single parameter $\lambda$; we consider a Gamma
distribution  for the prior on $\lambda$.

\medskip
  To highlight our result presented in  Section \ref{sec:posterior}, we consider the following strategy.  For $\sigma=1$, we first compute what we call the ``true" marginal likelihood $P_\bY(\by)$ (red line in Figure \ref{fig.ExaLog1}), using the  explicit solution of (\ref{eq:logistic}) and numerical integration. In a second step, we approximate the solution of (\ref{eq:logistic}) by the Euler scheme (equivalent to the Runge-Kutta solver of order 1), for various step sizes  $h_k$.  The  marginal likelihood  $P_\bY^{h_k(1)}(\by)$ is computed using both numerical integration and the Monte Carlo strategy presented in Section \ref{subset:comp Bayes factor}.  These results are plotted in black in Figure \ref{fig.ExaLog1},  with solid thin lines for the numerical integration and triangles for the Monte Carlo computation. 
 In a third step, the same strategy is applied but using a Runge-Kutta solver of order $2$ (in green in Figure \ref{fig.ExaLog1}). Finally, we use a classical Runge-Kutta solver  of order $4$ (in green in Figure \ref{fig.ExaLog1}).  In the end, for each order $p$, the estimated values $\hat{P}_\bY^{h_k( p)}(\by)$ are used to compute the regression functions $\hat{P}_\bY^{h_k( p)}(\by) = a + bh^p$ (thick lines in black, green and blue on Figure \ref{fig.ExaLog1}). These results are presented in Figures \ref{fig.ExaLog1} and \ref{fig.ExaLog2}. 

\medskip
  The same is done for $\sigma=30$ but only the RK solver of order 4 is considered. The equivalent results are presented in Figure  \ref{fig.ExaLog4}. 
The samples from the posterior distributions are obtained using a t-walk MCMC algorithm \citep{Christen&Fox2010}. 
Next we discuss some aspects of this numerical experiment.

\begin{figure}
\centering
\includegraphics[height=10cm, width=12cm]{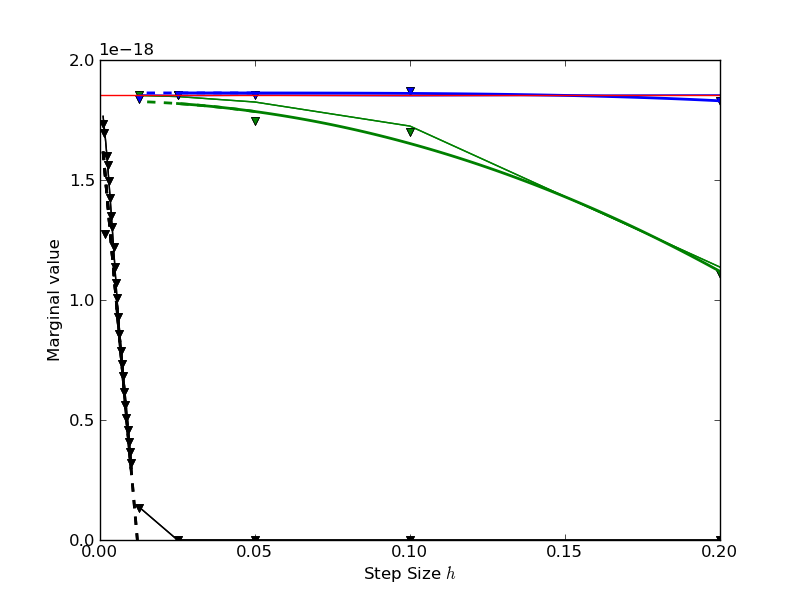} 
\caption{\emph{Study on synthetic data for the Logistic growth with $\sigma=1$}.   Marginal
$P_\bY^h(\by)$ for various step sizes, computed by numerical integration (solid thin lines) or  estimated using
the MCMC sample (triangles).  In black,   Runge-Kutta solver (RK) of order 1 (Euler), in green RK of order $p=2$,
in blue  RK of order $p=4$. Red line:  true marginal $P_\bY(\by)$ calculated using numerical integration on the analytic solution.
Thick lines indicate the regression for estimated values for $\hat{P}_\bY^h(\by) = a + b h^p$ for the orders $p=1,2,4$.}
\label{fig.ExaLog1} 
\end{figure}

\medskip
 $\bullet$  We would like to highlight that the thick lines and the triangles are quite similar, meaning that the  Monte Carlo strategy (derived as a by-product of the MCMC implementation) is an efficient solution to estimating the marginal likelihood in this context. This approximation has the great advantage of not to involving any new ODE solver runs and thus has a minimal computational cost.

\medskip
 $\bullet$  As predicted, the Euler scheme has a linear approximation regime to the correct marginal. To have any substantial save in CPU without compromising posterior inference precision we would need to have a very small step size, i.e. there is no `flat part' in order to take considerable larger step sizes.  On the contrary, the Runge-Kutta solvers of order 2 and specially the classical RK of order 4,
they indeed have a clear flat section where a nearly perfect estimation has been reached.  This allows for choosing a much larger step size, meaning a far coarser ODE numerical solver which still has basically no difference in the resulting inference and may be seen
in the resulting posterior distributions in Figures~\ref{fig.ExaLog2}(b) and~\ref{fig.ExaLog4}(c).

\medskip
 $\bullet$   For the RK or order $4$,  we perform a linear regression using the estimated $\hat{P}_\bY^{h_k}(\by)$  with $h_k= 0.2$,  $0.1$,  $0.05$,  $0.025$ for $\sigma=1$ and $h_k= 0.8$,  $0.6$,  $0.4$,  $0.2$,  $0.1$,  $0.05$ for $\sigma=30$. Using the formula given in Proposition \ref{prop:2},  we deduce an estimation (projection) of the exact marginal likelihood $\hat{P}_\bY(\by)$,  which has to be compared to the true value  $P_\bY(\by)$ (obtained using the exact solution of the ODE and a numerical integration). The results are given in the following
Table~\ref{tab.comp}.

\begin{table}
\begin{center}
\begin{tabular}{ccc}
\hline
$\sigma$ & $P_\bY(\by)$ & $\hat{P}_\bY(\by)$\\
\hline
$1$ & $1.854\, 10^{-18}$ & $1.862\, 10^{-18}$\\
$30$ & $1.638\, 10^{-60}$ & $1.699\, 10^{-60}$\\
\hline
\end{tabular}
\caption{\label{tab.comp}Comparison of exact an estimated marginals for the Ringue-Kutta method of order 4.}
\end{center}
\end{table}

 $\bullet$
We believe that the most important message is that, for the RK of order $4$, as soon as $h$ is lower than some threshold
(we took 0.05 for $\sigma=1$ and 0.1 for $\sigma=30$), the Bayes factor ratio $P_\bY^{h_k(2)}(\by) / P_\bY(\by)$
is greater than 0.99, making the models indistinguishable on the Jeffrey's Bayesian scale and leading to nearly identical
posterior distributions (see Figures~\ref{fig.ExaLog2}(b) and~\ref{fig.ExaLog4}(c)) for $\lambda$. However, the computational time required to estimate the parameters using the smallest $h$ explodes (see Figures~\ref{fig.ExaLog2}(a) and~\ref{fig.ExaLog4}(b)) from $2$ min for  $h=0.05$ to $17$ min for $h=0.01$, for $\sigma=1$, and from $2.5$ min for  $h=0.1$ to $36$ min for $h=0.000625$, for $\sigma=30$.

\begin{figure}
\begin{center}
\begin{tabular}{c l}
(a) 
\includegraphics[height=5cm, width=5cm]{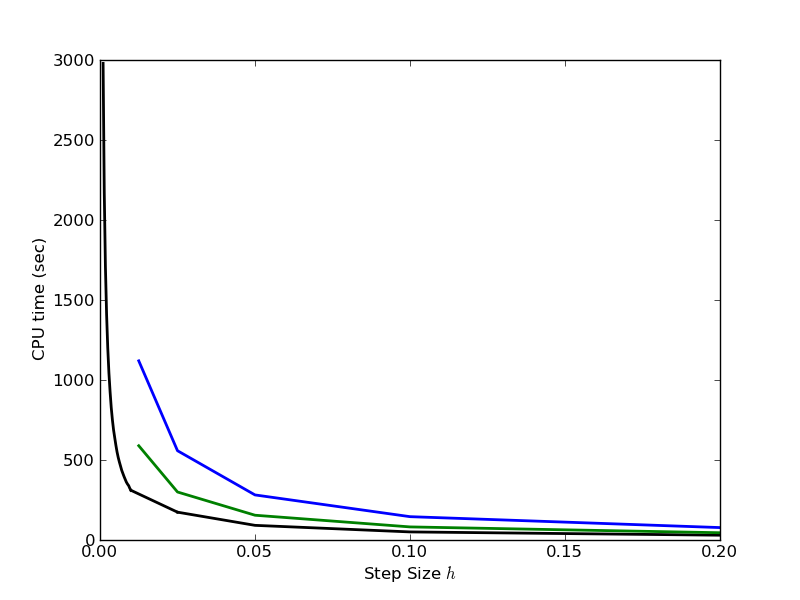} & 
(b) 
\includegraphics[height=5cm, width=5cm]{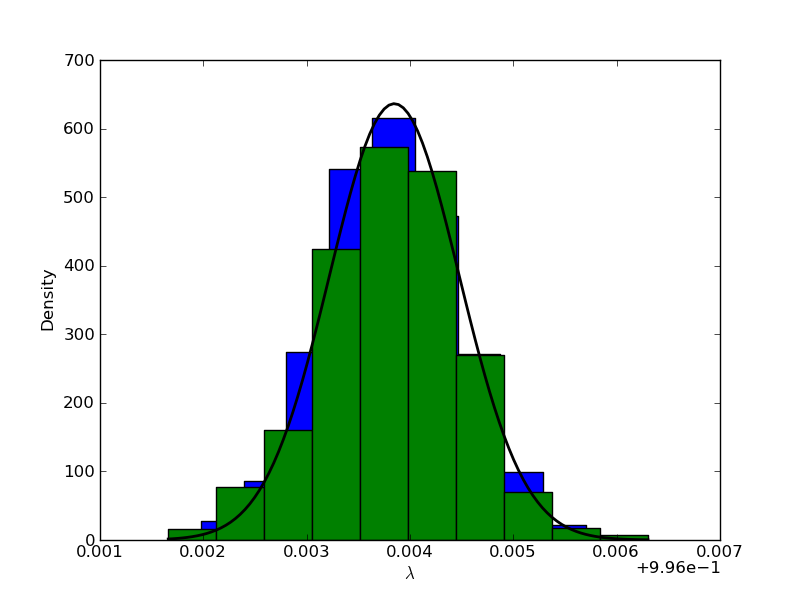} \\
\end{tabular} \\
\caption{(a) CPU time for various step values $h_k$ and $p=1,2,4$, relative to 10,000 iterations of the MCMC.  (b) Posterior distribution of $\lambda$ the for
RK4 solver, $p=4$, for step sizes $h=0.01$ and $h=0.06$ (histograms) and exact posterior (black density).  10,000 iterations
of the MCMC took 17 min for $h=0.01$ and 2 min for $h=0.05$; a 90\% reduction in CPU time with no noticeable difference
in the resulting posterior distribution.}
\label{fig.ExaLog2} 
\end{center}
\end{figure}

\smallskip  $\bullet$ Note that the ranges of considered values for $h_k$ are different for $\sigma=1$ and $\sigma=30$. This has to be linked to the remark we have made above:  the error induced by the numerical integration of the ODE can not be considered per-se but with respect to the observation noise. When $\sigma=30$,  the step size $h^*$ such that for any $h \leq h^*$ all the models $\M^h$ are equivalent on the Jeffrey's scale is much more higher involving even larger computational time savings. 

\begin{figure}
\begin{center}
\begin{tabular}{c l}
\multicolumn{2}{c}{(a)\includegraphics[height=7cm, width=10cm]{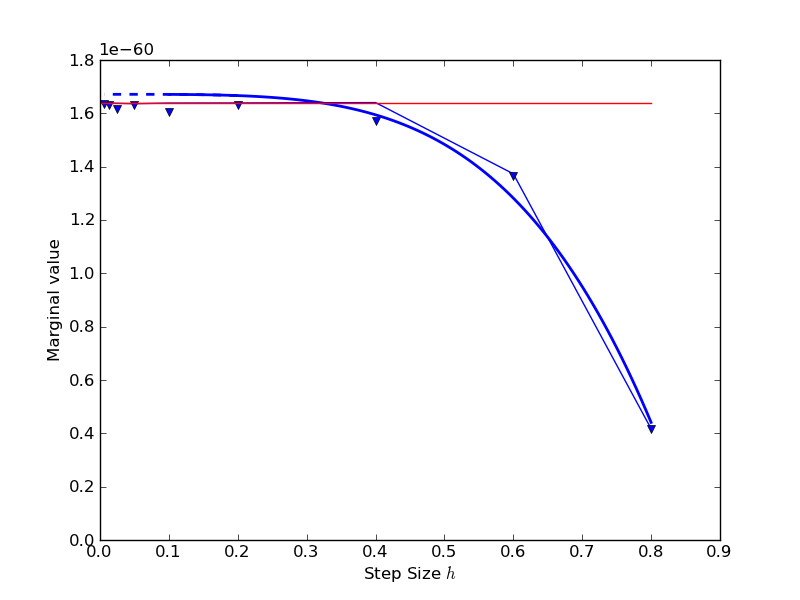}}\\ 
(b)\includegraphics[height=5cm, width=5cm]{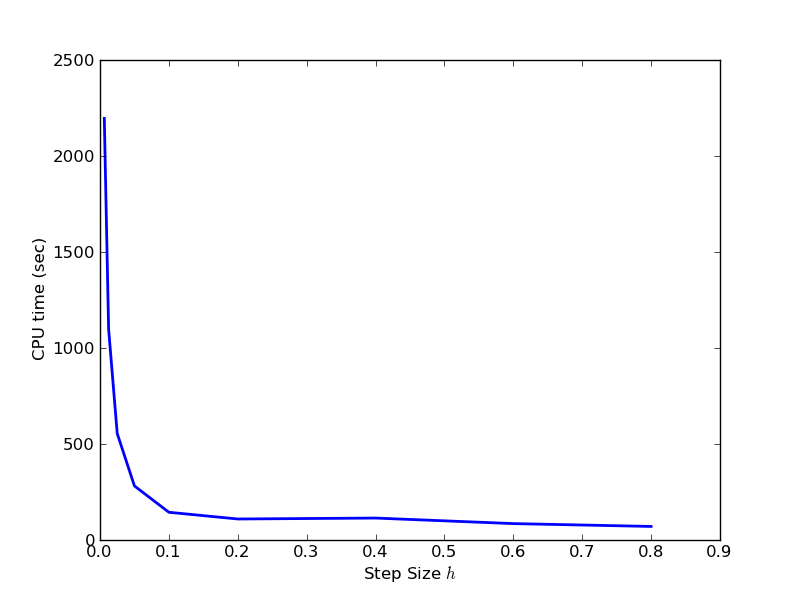}&
(c)\includegraphics[height=5cm, width=5cm]{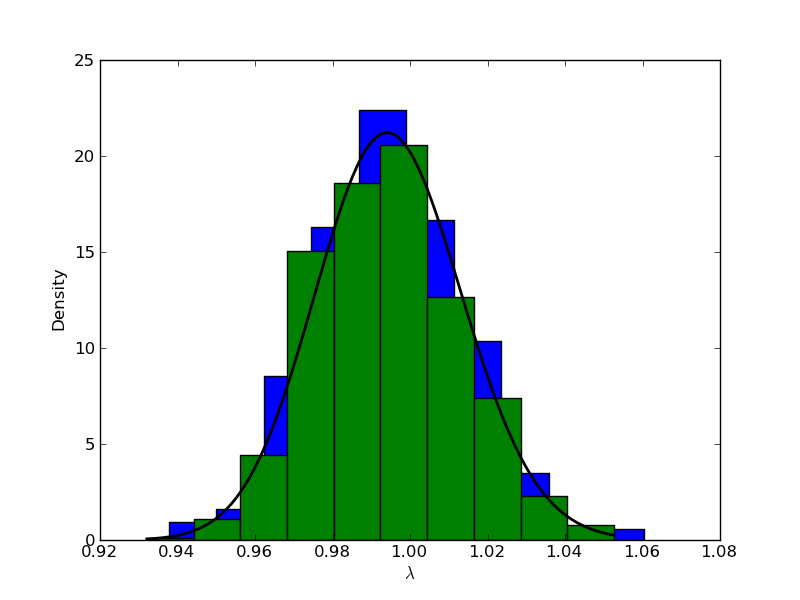} \\
\end{tabular} \\
\caption{\label{fig.ExaLog4} \emph{Study on synthetic data for the Logistic growth with $\sigma=30$}.  (a) Marginal
$P_\bY^h(\by)$ for various step sizes, both exact (solid thin lines, using numerical integration) and estimated using
the MCMC sample (triangles).  We use a Rungue-Kutta solver of order 4 (classical RK4, blue), only.
Red line:  true marginal $P_\bY(\by)$ calculated using numerical integration on the analytic solution.
Thick lines indicate the regression for estimated values for $\hat{P}_\bY^h(\by) = a + b h^p$ for the order $p=4$.
(b) Corresponding CPU time, relative to 10,000 iterations of the MCMC.  (c) Posterior distribution of $\lambda$ the for
RK4 solver, $p=4$, for step sizes $h=0.00625$ and $h=0.1$ (histograms; and exact posterior, black density).
The former takes 36 min and the latter 2.5 min.}
\end{center}
\end{figure}

\subsection{A Diabetes minimal model}\label{sec:exa2}

We know illustrate our results on a real dataset and a more complex model for an Oral 
Oral Glucose Tolerance Test (OGTT). After having briefly described the experiment and the model,
we present our results.

\medskip
 An Oral Glucose Tolerance Test (OGTT)
is performed for diagnosis of diabetes, methabolic syndrome and other conditions.
After a night sleep, fasting patients are measured for blood glucose and asked to drink a sugar
concentrate.  Blood glucose is then measured at the hour, two hours and sometime at three hours, depending on
local practices. We are developing a minimal model for blood glucose-insulin interaction base on a
two compartment model.  One simple
transfer compartment of glucose in the digestive system and one more complex compartment for blood glucose
and interactions with Insulin and other glucose substitution mechanisms.  Here we present this model
to show our methodology estimating one parameter only (namely, the insulin sensitivity). 
Although there is no analytical solution we are able to find the marginal $P_\bY^h(\bY)$ by numerical
integration (since there is only one parameter involved) for
comparison purposes.

\medskip
 Let $G(t)$ be blood glucose level at time $t$, in mg/dL.  Let $I(t)$ be blood insulin level at time $t$ and
$L(t)$ ``glucagon'' levels, to promote liver Glycogen glucose production, in arbitrary units.  Let $D(t)$ be the digestive system
`glucose level'; we take it as a compartment in which glucose is first stored (eg. stomach and digestive tract) and in turn
delivered into the blood stream (we state $D(t)$ in the same units as for $G(t)$ and therefore the mean life parameter $\theta_2$ in (\ref{eqn.minmod.D}) is the same as the one used in (\ref{eqn.minmod.G}) below). Let also $G_b$ be the glucose base line, (=80 mg/dL, fixed). Our model is described by the following system of differential equations 

\begin{eqnarray}
\label{eqn.minmod.G} \frac{dG}{dt} & = &    \left(L - I \right) G +  \frac{D}{\theta_2} , \\ 
\label{eqn.minmod.I}   \frac{dI}{dt}  & = &   \theta_0 \left( \frac{G}{G_b} - 1 \right)^+ - \frac{I}{a} , \\
\label{eqn.minmod.L}  \frac{dL}{dt} & = &   \theta_1 \left( 1 - \frac{G}{G_b} \right)^+ - \frac{L}{b} , \\
\label{eqn.minmod.D} \frac{dD}{dt} & = & - \frac{D}{\theta_2} .
\end{eqnarray}

\medskip
 A brief explanation of the model goes as follows.  When glucose goes above the normal threshold
$G_b$, Insulin is produced,
ie. its derivative increases, see (\ref{eqn.minmod.I}).  This, in turn, acts on blood glucose to decrease its concentration;
a mass-action type term is introduced in (\ref{eqn.minmod.G}) to decrease the derivative of $G$.
$L$ is an abstract term related to the glucose recovery system.  When Glucose $G(t)$ goes below the
normal threshold ($G_b$) $L$ increases, see (\ref{eqn.minmod.L}), to increase the derivative of $G(t)$
(thus eventually increasing the glucose), see (\ref{eqn.minmod.G}).  Finally, $D(t)$ represents the glucose in
the digestive compartment that will be transferred to the blood stream, see (\ref{eqn.minmod.D}) and
(\ref{eqn.minmod.G}).  We analyze data from an OGTT conducted in an obese male adult patient
with a suspected methabolic syndrome condition; the corresponding data are plotted on Figure~\ref{fig.ExaDiabMinMod.data}.
All parameters are positive and will be set to
$\theta_1 = 26.6, \theta_2 = 0.2, a=1, b=2$, while $\theta_0$ is taken as unknown and will be estimated from data.
This is an unusual experimental data set in which Glucose was measured every 30 min
up to 2 hr.

\begin{figure}
\begin{center}
\includegraphics[height=8cm, width=12cm]{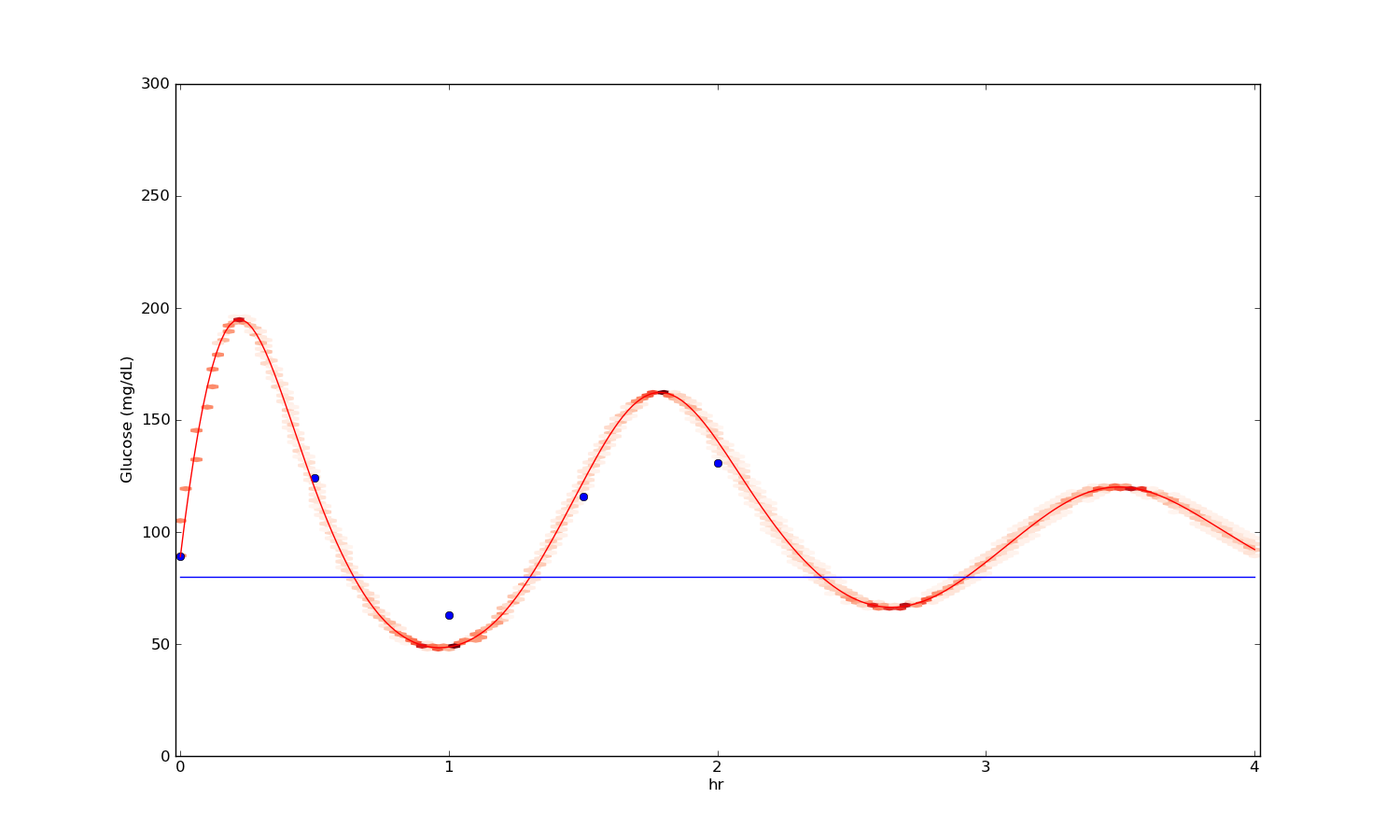}
\caption{\label{fig.ExaDiabMinMod.data} OGTT test performed in an obese male adult, with glucose measurents
taken every 30 min up to 2 hr.  Note the oscilating nature of the data, typically of a not well control Insulin-Glucose
system.  Both $\theta_0$ and $\theta_1$ have large values in comparison to normal subjects.  The MAP model
is shown in red, along with draws from the posterior predictive distribution shown in the shaded areas.}
\end{center}
\end{figure}

\medskip
 Our Bayesian inference is performed as follows.  We have observations $d_1, d_2, \ldots, d_n$
for Glucose, thus we let
$$
d_i = G(t_i) + e_i ~~\text{where}~~e_i \sim N(0, \sigma ),
$$
and $G(0) = d_0$ the initial condition; we fix the measurement error to $\sigma = 5$ (the observation functional
is therefore $f(X) = X_1$).  From this a likelihood is constructed.
A Gamma prior distribution is assumed for parameter $\theta_0$, with shape parameter $=5$ and rate $=2/5$, thus with
mean $=2$, apparently a value for a normal person.  Using an order 4 Rugue-Kutta solver with varying step size
we perform a MCMC for these parameters using the t-walk (Christen and Fox, 2010).

\begin{figure}
\begin{center}
\begin{tabular}{c l}
\multicolumn{2}{c}{(a)\includegraphics[height=7cm, width=10cm]{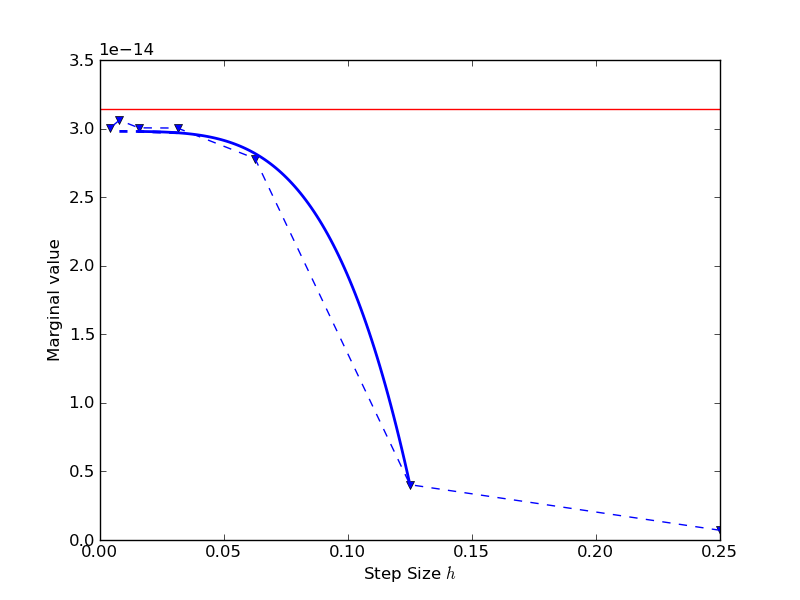}}\\ 
(b)\includegraphics[height=5cm, width=5cm]{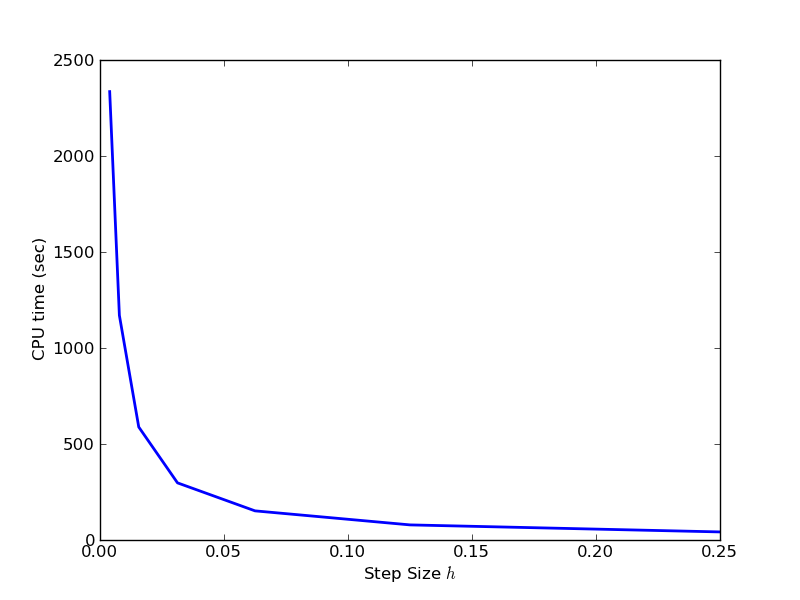}&
(c)\includegraphics[height=5cm, width=5cm]{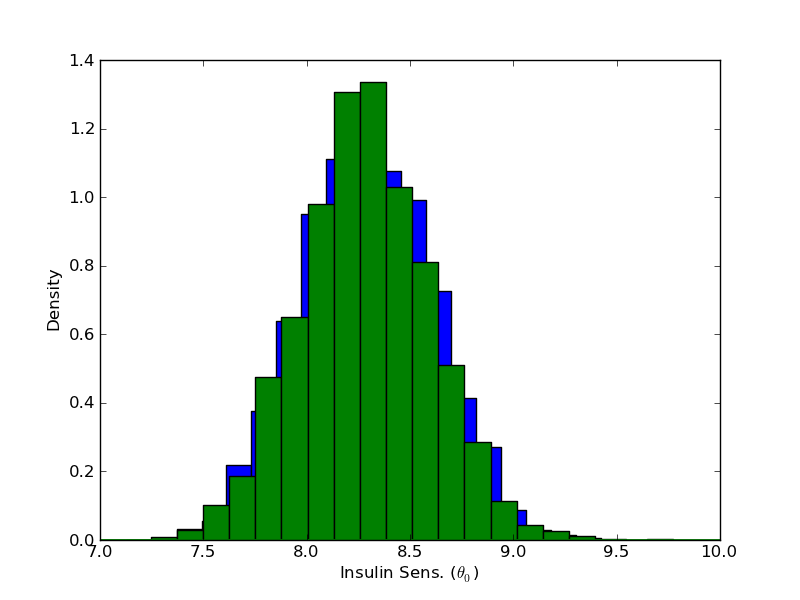} \\
\end{tabular} \\
\caption{\label{fig.ExaDiabMinMod} Bayes Factors study for the Diabetes minimal model.  An order 4 Runge-Kutta solver
was used to produce marginal values $P_\bY^h(\bY)$ for step sizes as show in (a) and their corresponding CPU times
are depicted in (b).  The red line in (a) is the numerical
integration approximation of $P_\bY^h(\bY)$ using step size $0.25 \cdot 2^{-7}$ (smallest step size used) while the triangles
are Monte Carlo estimates performed as in (\ref{eqn.estMarg}); these seem to slightly underestimate the former.  The solid blue line is a regression model $a + b h^4$ estimate using step sizes marginal estimates from
$0.25 \cdot 2^{-1}$ to $0.25 \cdot 2^{-4}$ only. 
In (c) we compare the resulting posterior with step size $0.25 \cdot 2^{-3}$ and $0.25 \cdot 2^{-7}$
showing basically no difference and resulting in a near 90\% reduction in CPU evaluation time.}
\end{center}
\end{figure}

To make the solver evaluation time steps (in hr) include the observation times we take a rough time step of 15min (=0.25hr)
and divide it into finer time steps defined as $h^{(k)} = 0.25 \cdot 2^{-k}$.  Our experiments included $k=0, \ldots , 7$,
as seen in in Figure~\ref{fig.ExaDiabMinMod}(a).  We only use an order 4 Runge-Kutta solver,
resulting in the 4th grade polynomial regression and a flat section already at $h^{(3)} = 0.25  \cdot 2^{-3} = 1.875$min.
If compared with our minimum time step
of $h^{(7)} = 0.25  \cdot 2^{-7} = 7$sec, the resulting CPU time of the MCMC is more than 90\% larger.  However,
the resulting posterior distributions for $h^{(3)}$ and $h^{(7)}$ are basically identical (see Figure~\ref{fig.ExaDiabMinMod}(c)).
The estimated marginal is $P_{\bY}^{h^{(3)}}(\by) \approx 3\cdot 10^{-14}$, calculated using the MCMC samples and
(\ref{eqn.estMarg}), while $P_{\bY}^{h^{(7)}}(\by) \approx 3.2\cdot 10^{-14}$
calculated using numerical integration (red line in Figure~\ref{fig.ExaDiabMinMod}(a)).

\section{Discussion}\label{sec:discussion}

We advance on some theoretical aspects of the Bayesian analysis of ODE systems.  As opposed to more standard
(Bayesian) statistical analyses, Inverse Problems present the added difficulty that the regressor function is not
analytically tractable and numerical approximations need to be used.  In general, the replacement of the theoretical (non-available) solution of the differential system by a numerical approximation is ignored and the solver being used as a black box.
However, recently, research has been directed at trying to quantify the consequences of such an approximation,
commonly by comparing expected values of the resulting posterior distributions,
like the exact vs the numerical Posterior means. 

In this paper we adopt a different approach, basing our comparison on the use of Bayes Factors,
which is the natural tool to comparing models in a Bayesian context. 
There are still some particular issues to be solve when applying our results to more realistic inverse problems like
estimating the marginals in a multidimensional parameter problem and analyzing stiff problems were a multistep
method would need to be used.  However, we may highlight the following remarks. 

First, we contribute to the intuitive idea that the ODE solver
approximation error should be put in the perspective of the observational error.  Bayes Factors, and the
Bayesian model comparison machinery, can be used as an appropriate measure of the solver accuracy,
precisely in the perspective of the observational error considered in the model.  Result~\ref{prop:2} establishes
a consistency in order accuracy for the solver and for the posterior distribution, considering BF's. As a consequence,
numerical solver precision may be viewed in this perspective and not solely as a black box regressor.
As far as the main aim is to make inference on parameters, there is no need to use to highest precision
if the data are contaminated by a non-neglectable quantity of noise. In a domain where the computational
time is important,  we have proved that considerable time savings can be done, only by using a reasonable step size in the solver. 

Secondly, we show how the BF may be approximated even in this scenario where the exact model is not
available.  This result is of particular interest, since it allows to compare the accuaracy of our approximate posterior
without being able to work on the theoretical model directly.
The computation of marginal likelihoods is an important topic in the Bayesian literature. In this paper,
we propose the use of the Gelfand and Dey's estimator, which has the great advantage of not requiring any
additional numerical evaluation of the differential system, after the MCMC was performed. However, we are aware that
the Gelfand and Dey's estimator may be highly unstable when the dimension of the parameters increases.
The use of a Kernel Density Estimate weighting
function in (\ref{eqn.estMarg}) can help to stabilize the estimate but is not a universal solution.
If the dimension of the parameters increases, other strategies should be considered, still keeping in mind that any
additional numerical evaluation of the ODE system may have a considerable computational cost. 
Our results would also need to be stated for multiple dimension observation functions $f$; we leave
these considerations for future research.

\section{Acknowledgments}

We thank Dr Silvia Quintana for kindly providing the OGTT data for example in~\ref{sec:exa2}.
MAC and JAC would like to acknowledge financial support 
from Fondo Mixto de Fomento a la Investigaci\'on Cient\'{\i}fica y Tecnol\'ogica, 
CONACYT-Gobierno del Estado de Guajanuato, GTO-2011-C04-168776.

\bibliographystyle{gSCS.bst}
\bibliography{BayesInvProbReview}

\end{document}